\documentclass[letterpaper, 10 pt, conference]{ieeeconf}  
\IEEEoverridecommandlockouts                              % This command is only needed if
\usepackage{graphics} % for pdf, bitmapped graphics files
\usepackage{epsfig} % for postscript graphics files 
\usepackage{amsmath,amssymb,amsfonts}
\usepackage{graphicx}  
\usepackage{tikz} 
\usepackage{epstopdf}
\usepackage{cite}
\usepackage{algorithm,algorithmic,setspace}
\usepackage{url}
\usepackage{soul}
\usepackage[colorlinks=true,linkcolor=blue, citecolor=blue, urlcolor=blue]{hyperref}
\allowdisplaybreaks

\newcommand{\doi}[1]{\href{https://doi.org/#1}{doi: #1}}   

\newtheorem{thm}{\bf Theorem}
\newtheorem{lem}{\bf Lemma}
\newtheorem{rem}{\bf Remark}
\newtheorem{prop}{\bf Proposition}
\newtheorem{assum}{\bf Assumption}
\newtheorem{defn}{\bf Definition}

\DeclareMathOperator{\tr}{tr}
\DeclareMathOperator{\dom}{dom}  
\DeclareMathOperator{\diag}{diag}

\newcommand{\T}{^{\top}} 
\newcommand{\ie}{\textit{i.e.}}
\allowdisplaybreaks

\title{\LARGE \bf 
Nonlinear Attitude Estimation Using Intermittent Linear Velocity\\ and Vector Measurements 
}

\author{Miaomiao Wang and Abdelhamid Tayebi % <-this % stops a space
	\thanks{This work was supported by the National Sciences and Engineering Research Council of NSERC-DG RGPIN-2020-06270. }
	\thanks{The authors are with the Department of Electrical and Computer Engineering, Western University, London, ON N6A 3K7, Canada. A. Tayebi is also with the Department of Electrical Engineering, Lakehead University, Thunder Bay, ON P7B 5E1, Canada.
		{\tt\small mwang448@uwo.ca}, {\tt\small atayebi@lakeheadu.ca}}%
} 

\begin{document}

\maketitle
\thispagestyle{empty}
\pagestyle{empty}
%\parskip -0.05in
	
%%%%%%%%%%%%%%%%%%%%%%%%%%%%%%%%%%%%%%%%%%%%%%%%%%%%%%%%%%%%%%%%%%%%%%%%%%%%%%%%
\begin{abstract}
	 This paper investigates the problem of continuous attitude estimation on $SO(3)$ using continuous angular velocity and linear acceleration measurements as well as intermittent linear velocity and inertial vector measurements. First, we propose a nonlinear observer for the case where all the measurements are continuous  and almost global asymptotic stability (AGAS) is shown using the notion of almost global input-to-state stability (ISS) on manifolds. Thereafter, a hybrid attitude observer, with AGAS guarantees, is proposed in terms of intermittent linear velocity and vector measurements.  Numerical simulation results are presented to illustrate the performance of the proposed hybrid observer.
\end{abstract}

\section{Introduction}    
The algorithms used for the determination of the orientation, or attitude, of a rigid body system, are instrumental in robotics and aerospace applications. The attitude can be determined through the integration of the angular velocity which is not a viable solution in practice due to the integral drift over time due to measurement bias and noise. In the early 1960s, many static attitude determination techniques, relying on body-frame observations of some vectors known in the inertial frame, have been introduced (see, for instance, \cite{Wahba,shuster1979}). These vector observations can be obtained using different types of sensors such as low-cost inertial measurement unit (IMU) sensors (including an accelerometer, a gyroscope and a magnetometer), or sophisticated sensors such as sun sensors and star trackers. However, these static attitude determination algorithms, although simple, do not perform well in the presence of measurement noise. This motivated their reinforcement with Kalman-type filters leading to dynamic attitude estimation algorithms (see the survey paper \cite{Crassidis2007survey}). Although successfully implemented in many practical applications,  
these Kalman-based dynamic estimation techniques rely on local linearizations (approximations) and lack rigorous stability analysis in the global sense.

Recently, a class of geometric nonlinear attitude observers, evolving on the Special Orthogonal group $SO(3)$, have made their appearances in the literature. These geometric observers take into account the topological properties of group $SO(3)$ and provide AGAS guarantees, \ie, the estimated attitude converges asymptotically to the actual one from almost all initial conditions except from a set of zero Lebesgue measure (see, for instance, \cite{mahony2008nonlinear}). 
Due to the space topology of $SO(3)$, AGAS is the strongest result one can achieve via time-invariant smooth observers. To strengthen the stability results, hybrid observers for global attitude estimation have been considered in \cite{wu2015globally,berkane2017hybrid}. On the other hand, in many low-cost applications, most of the existing attitude estimation techniques rely on IMU  measurements and assume  negligible linear accelerations. This small-acceleration assumption allows to use the gravity vector as one of the inertial vectors measured in the body frame via an accelerometer. In applications involving non-negligible linear accelerations, one can use the so-called velocity-aided attitude observers that rely on IMU measurements and the linear velocity in either the inertial frame \cite{roberts2011attitude,grip2015globally,berkane2017attitude} or the body frame \cite{bonnabel2008symmetry,hua2016stability}.

 From the implementation point of view, attitude estimation often involves different types of sensors with different sampling rates. For instance, the measurements of a global positioning system (GPS)  and a vision system are obtained at much lower rates than the IMU measurements. However, most of the existing attitude observers are designed based on the assumption of continuous output measurements. It is clear that the stability and performance would be altered if one tries to implement these continuous-time observers with intermittent measurements. In this context, some more recent results dealing with discrete measurements have been considered, for instance, the discrete-time attitude observers proposed in \cite{barrau2015intrinsic,bhatt2020rigid} and the continuous-discrete attitude observers proposed in \cite{khosravian2015recursive,berkane2019attitude}. The latter category assumes that the high-rate measurements of the angular velocity are continuous and the low-rate measurements of inertial vectors are intermittent. A predictor-observer approach has been proposed in \cite{khosravian2015recursive} based on a cascade combination of an output predictor and a continuous attitude observer. The output predictor was designed to smooth the vector measurements through a forward integration on $SO(3)$ of the continuous angular velocity measurements. In \cite{berkane2019attitude}, the authors consider a predict-update hybrid approach, where the estimated attitude is continuously updated by integrating the continuous angular velocity and discretely updated through jumps upon the arrival of the intermittent vector measurements. 
 
In this paper, we consider the problem of continuous attitude estimation using continuous (high-rate) angular velocity and linear acceleration measurements and intermittent (low-rate) linear velocity and inertial vector measurements. We first propose a continuous-time velocity-aided attitude observer on $SO(3)\times \mathbb{R}^6$ with AGAS guarantees relying on the notion of almost global ISS on manifolds. Then, motivated by the work \cite{wang2020nonlinear}, we propose a hybrid velocity-aided attitude observer in terms of intermittent linear velocity and inertial vector measurements with AGAS guarantees. In particular, all the estimated states are continuously updated through integration using the continuous angular velocity and linear acceleration measurements, and discretely updated upon the arrival of the intermittent linear velocity and vector measurements. The proposed hybrid observer has a similar structure as \cite{berkane2019attitude}, while the estimated attitude from our hybrid observer is continuous without any additional smoothing algorithm. The fact that our proposed hybrid observer generates continuous estimates of the attitude makes it suitable for practical applications involving observer-controller implementations.

%The remain of this paper is organized as follows:

%%%%%%%%%%%%%%%%%%%%%%%%%%%%%%%%%%%%%%%%%%%%%%%%%%%%%%%%%%%%%%%%%%%%%%%%%%%%%%%%
\section{Preliminaries} \label{sec:backgroud}
\subsection{Notations and Definitions}
The sets of real, non-negative real, natural and positive natural numbers are denoted by $\mathbb{R}$, $\mathbb{R}_{\geq 0}$, $\mathbb{N}$ and $\mathbb{N}_{>0}$, respectively. We denote by $\mathbb{R}^n$ the $n$-dimensional Euclidean space and $\mathbb{S}^{n-1}$ the set of unit vectors in $\mathbb{R}^{n}$. The Euclidean norm of a vector $x\in \mathbb{R}^n$ is defined as $\|x\| = \sqrt{x\T x}$. Let $I_n$ denote the  $n$-by-$n$ identity matrix, $0_{n}$ denote the $n$-by-$n$ zero matrix, and $0_{n \times m}$ denote the $n$-by-$m$ zero matrix.  
For a given symmetric matrix $A\in \mathbb{R}^{n\times n}$, we define $\mathcal{E}(A)$ as the set of all unit-eigenvectors of $A$, and $\lambda_m(A)$ and $\lambda_M(A)$ as the minimum and maximum eigenvalues of $A$, respectively. 
Given two matrices, $A,B\in \mathbb{R}^{m\times n}$, their Euclidean inner product is defined as $\langle\langle A,B\rangle\rangle = \tr(A\T B)$ and the Frobenius norm of $A$ is defined as $\|A\|_F = \sqrt{\langle\langle A,A\rangle\rangle}$. For each $x=[x_1,x_2,x_3]\T \in \mathbb{R}^3$, we define $x^\times$ as a skew-symmetric matrix given by
\[
x^\times = \begin{bmatrix}
	0 & -x_3 & x_2 \\
	x_3 & 0 & -x_1 \\
	-x_2 & x_1 & 0
\end{bmatrix} 
\] 
and $\text{vec}(\cdot)$ as the inverse operator of the map $(\cdot)^\times$, such that $\text{vec}(x^\times) = x$.	 
For a matrix $A\in \mathbb{R}^{3\times 3}$, we denote   $\mathbb{P}_a(A) := \frac{1}{2}(A-A\T)$ as the anti-symmetric projection of $A$. Define the composition map $\psi: =\text{vec} \circ \mathbb{P}_a $  such that, for a matrix $A=[a_{ij}] \in \mathbb{R}^{3\times 3}$, one has
$$
\psi (A) := \text{vec} (\mathbb{P}_a(A)) =\frac{1}{2}[
a_{32} - a_{23},
a_{13} - a_{31},
a_{21} - a_{12}
]\T.
$$
For any $A\in \mathbb{R}^{3\times 3}, x\in \mathbb{R}^3$, one can verify that
$
\langle\langle A, x^\times\rangle\rangle 
= 2x\T \psi (A)
$. The 3-dimensional  \textit{Special Orthogonal group} is denoted by  
$$
SO(3): = \left\{R\in \mathbb{R}^{3\times 3}| R\T R = I_3, \det(R) = +1\right\}.
$$
The \textit{Lie algebra} of $SO(3)$, denoted by $\mathfrak{so}(3)$ is given by
$$
\mathfrak{so}(3) := \left\{\Omega\in \mathbb{R}^{3\times 3}| \Omega\T=-\Omega \right\}.
$$	 
For any $R\in SO(3)$, we define $|R|_I\in [0,1]$ as the normalized Euclidean distance on $SO(3)$ with respect to the identity $I_3$, which is given by $|R|_I^2 = \tr(I_3-R)/4$. 
Let the map $\mathcal{R}_a: \mathbb{R} \times \mathbb{S}^2\to SO(3)$ represent the well-known angle-axis parameterization of the attitude defined by
$$
\mathcal{R}_a(\theta,u) :=  I_3 + \sin(\theta) u^\times + (1-\cos(\theta))(u^\times)^2
$$
with $\theta\in \mathbb{R}$ denoting the
rotation angle and $u\in \mathbb{S}^2$ denoting the rotation axis. The following identity will be used throughout this paper
\begin{equation} 
	\psi(Q X) = \frac{1}{2} \sum_{i=1}^m \rho_i (X\T r_i) \times  r_i, ~ \forall X\in SO(3)   \label{eqn:psiAR}
\end{equation} 
where $m\in \mathbb{N}_{>0}$, $ \rho_i\in \mathbb{R}, r_i\in \mathbb{R}^3$ for all $i\in\{1,2,\dots,m\}$ and $Q=\sum_{i=1}^m \rho_i r_ir_i\T$.  %{\color{blue} In this paper, we make use of the framework of hybrid dynamical systems presented in \cite{goebel2009hybrid,goebel2012hybrid,teel2013lyapunov}.}

\subsection{Hybrid Systems Framework} 
Consider a  smooth manifold $\mathcal{M}$ embedded in $\mathbb{R}^n$, and  let $T \mathcal{M}:=\bigcup_{x\in \mathcal{M}} T_x \mathcal{M}$ 
denote its tangent bundle. A general model of a hybrid system is given as \cite{goebel2012hybrid}:
\begin{equation}\mathcal{H}:  
	\begin{cases}
		\dot{x} ~~= F(x),& \quad x \in \mathcal{F}   \\
		x^{+} \in G(x),& \quad x \in \mathcal{J}
	\end{cases} \label{eqn:hybrid_system}
\end{equation}
where $x\in \mathcal{M}$ denotes the state, $x^+$ denotes the state after an instantaneous jump, the \textit{flow map} $F: \mathcal{M} \to T \mathcal{M}$ describes the continuous flow of $x$ on the \textit{flow set} $\mathcal{F} \subseteq \mathcal{M}$, and the \textit{jump map} $G: \mathcal{M}\rightrightarrows  \mathcal{M}$ (a set-valued mapping from $\mathcal{M}$  to $\mathcal{M}$) describes the discrete flow of $x$ on the \textit{jump set} $\mathcal{J} \subseteq \mathcal{M}$.  
A solution $x$ to $\mathcal{H}$ is parameterized by $(t, j) \in \mathbb{R}_{\geq 0} \times \mathbb{N}$, where $t$ denotes the amount of time passed and $j$ denotes the number of discrete jumps that have occurred.  A subset $\dom x \subset \mathbb{R}_{ \geq 0} \times \mathbb{N}$ is a \textit{hybrid time domain} if for every $(T,J)\in \dom x$, the set, denoted by $\dom x \bigcap ([0,T]\times \{0,1,\dots,J\})$, is a union of finite intervals  of the form $\bigcup_{j=0}^{J} ([t_j,t_{j+1}] \times \{j\})$ with  a time sequence  $0=t_0 \leq t_1 \leq \cdots \leq t_{J+1}$. 
A solution $x$ to $\mathcal{H}$ is said to be  \textit{maximal} if
it cannot be extended by flowing nor jumping, and \textit{complete} if its domain $\dom x$ is unbounded. Let $|x|_{\mathcal{A}}$  denote the distance of a point $x$ to a closed set $\mathcal{A} \subset \mathcal{M}$, 
and then the set $\mathcal{A}$ is said to be: 
\textit{stable} for $\mathcal{H}$ if for each $\epsilon>0$ there exists $\delta>0$ such that each maximal solution $x$ to $\mathcal{H}$ with $|x(0,0)|_{\mathcal{A}} \leq \delta$ satisfies $|x(t,j)|_{\mathcal{A}} \leq \epsilon$ for all $(t,j)\in \dom x$; \textit{globally attractive} for $\mathcal{H}$ if  every maximal solution $x$ to $\mathcal{H}$  is complete and satisfies $\lim_{t+j\to \infty}|x(t,j)|_{\mathcal{A}} = 0$ for all $(t,j)\in \dom x$; \textit{globally asymptotically stable} (GAS) if it is both stable and globally attractive for $\mathcal{H}$. 
Moreover, the $\mathcal{A}$ is said to be   \textit{exponentially stable} for   $\mathcal{H}$ if there exist $\kappa, \lambda>0$ such that, every maximal solution  $x$  to $\mathcal{H}$ is complete and satisfies  $|x(t,j)|_{\mathcal{A}} \leq \kappa e^{-\lambda(t+j)}|x(0,0)|_{\mathcal{A}}$ for all $(t,j)\in \dom x$ \cite{teel2013lyapunov}.
We refer the reader to \cite{goebel2012hybrid} and references therein for more details on hybrid dynamical  systems.

\subsection{AGAS and Almost Global ISS}

Let $\mathcal{M}$ be a smooth manifold.  Consider the following general nonlinear system on the manifold $\mathcal{M}$:
\begin{equation}
	\dot{x} = f(x,u) \label{eqn:nonliear_system}
\end{equation}
where $x\in \mathcal{M}$ is the state, $u\in \mathcal{U} \subset \mathbb{R}^m$ is the input, and $f:\mathcal{M}\times \mathcal{U}\to T\mathcal{M}$ is a locally Lipschitz manifold map which satisfies $f(x,u)\in T_x \mathcal{M}$ for all $x\in \mathcal{M},u\in \mathcal{U}$. The system \eqref{eqn:nonliear_system} (with $u\equiv 0$) is said to be AGAS at an invariant compact set $\mathcal{A}\subset \mathcal{M}$ if the set $\mathcal{A}$ is stable and the state $x$ tends to the set $\mathcal{A}$ from any initial conditions in $\mathcal{M}$ except a set of zero Lebesgue measure. Throughout this paper, we will make use of the notion of almost global ISS in \cite{angeli2004almost}.
\begin{defn}\label{def:A_ISS} 
	System \eqref{eqn:nonliear_system} is  almost globally ISS with respect to the set $\mathcal{A}$, if $\mathcal{A}$ is locally asymptotically stable for \eqref{eqn:nonliear_system} with $u\equiv 0$ and there exists  $\gamma \in \mathcal{K}$ such that for each locally essentially bounded and measurable input $u: \mathbb{R}_{\geq 0} \to \mathcal{U}$, there exists a zero Lebesgue measure
	subset $\aleph_{u} \subset \mathcal{M}$ such that 
	\begin{align}
		\limsup_{t\to +\infty} |x(t,x_0,u )|_{\mathcal{A}} \leq \gamma (\|u\|_{\infty}), ~~\forall x_0\in \mathcal{M}\setminus\aleph_{u}.
	\end{align}
\end{defn} 
The following lemma, adopted from \cite[Theorem 2]{angeli2004almost}, provides AGAS for a nonlinear cascaded system consisting of an almost globally ISS system and a globally asymptotically stable (GAS) system.
\begin{lem}\label{lem:AGAS_conti}
	Consider the following cascaded system:
	\begin{subequations} \label{eqn:cascaded_system}
		\begin{align}
			\dot{x}  &= f(x,y)  \label{eqn:dot_x}\\
			\dot{y}  &= g(y) \label{eqn:dot_y}  
		\end{align}
	\end{subequations}
	where $(x,y) \in \mathcal{M} \times \mathcal{N}$,   $f:\mathcal{M} \times \mathcal{N}\to T\mathcal{M}$ and $g: \mathcal{M} \times \mathcal{N} \to T\mathcal{N}$ are  locally Lipschitz with $f(x,y)\in T_x \mathcal{M}$ and $g(y)\in T_y \mathcal{N}$ for all $(x,y) \in \mathcal{M} \times \mathcal{N}$.  
	Suppose that
	\begin{itemize}	 				
		\item [1)] the $x$-subsystem  is almost globally ISS with respect to $\mathcal{A}_x \subset \mathcal{M}$ and input $y$,  
		\item [2)] the $y$-subsystem  is GAS at $\mathcal{A}_y \subset  \mathcal{N}$
	\end{itemize}
	Then, the cascaded system \eqref{eqn:cascaded_system} is AGAS at $\mathcal{A}:=\mathcal{A}_x \times \mathcal{A}_y$.
\end{lem}
Note that the cascaded system in Lemma \ref{lem:AGAS_conti} is a special case of that in \cite[Theorem 2]{angeli2004almost}, since $y$-subsystem \eqref{eqn:dot_y} is GAS instead of AGAS as in \cite[Theorem 2]{angeli2004almost}. Motivated by Lemma \ref{lem:AGAS_conti}, the following lemma provides AGAS for a cascaded hybrid system.
\begin{lem}\label{lem:AGAS_hybrid}
	Consider the following cascaded hybrid system
	\begin{align} \label{eqn:cascaded_hybrid_system}
		\arraycolsep=2.4pt\def\arraystretch{1.1}
		\underbrace{\begin{array}{ll}
				\dot{x} & = f(x,y)\\
				\dot{y} & = g(y) 
		\end{array}}_{(x,y)\in \mathcal{F}} ~
		\underbrace{\begin{array}{ll}
				x^+ &=  x \\
				y^+ &\in g'(y) 
		\end{array}}_{(x,y)\in \mathcal{J}}  
	\end{align}
	where $  (x,y) \in \mathcal{M} \times \mathcal{N}$, the functions $f$ and $g$ are described as per Lemma \ref{lem:AGAS_conti} and the   map $g': \mathcal{N} \rightrightarrows \mathcal{N}$. Suppose that the hybrid system \eqref{eqn:cascaded_hybrid_system} satisfies the hybrid basic conditions and  
	\begin{itemize}	
		\item [1)] the $x$-subsystem   is almost globally ISS with respect to $\mathcal{A}_x \subset \mathcal{M}$ and input $y$,
		\item [2)] the $y$-subsystem  is GAS at $\mathcal{A}_y \subset  \mathcal{N}$,
		\item [3)] every maximal solution to \eqref{eqn:cascaded_hybrid_system} is complete and $t\to +\infty$.
	\end{itemize}
	Then, the cascaded hybrid system \eqref{eqn:cascaded_hybrid_system} is AGAS at $\mathcal{A}$.
\end{lem} 
	The proof of Lemma \ref{lem:AGAS_hybrid} can be easily obtained from the proof of \cite[Theorem 2]{angeli2004almost} providing that $y$-subsystem is GAS at $\mathcal{A}_y$ with $\limsup_{t\to + \infty}|y(t,j)|_{\mathcal{A}_y} =0$ from item 2) and 3).

\section{Problem Statement}\label{sec:problem}

The kinematics   of a rigid body on $SO(3)$ are given by
\begin{align}
	\dot{R} &= R\omega^\times \label{eqn:dynamics_R}  
\end{align}
where $R\in SO(3)$ denotes the attitude of the rigid body, and $\omega\in \mathbb{R}^3$ denotes the angular velocity of the rigid body expressed in body  frame.
 
The measurement model of the linear acceleration $a$ obtained, for instance, from accelerometer, is given as
\begin{equation}
	a =  R\T (\dot{v}-g) \label{eqn:def_a} %\omega^\times v -R\T g + \dot{v} 
\end{equation}
where $v$ denotes the linear velocity expressed in the inertial frame, and $g$ is the gravity vector known in the inertial frame. We assume that the body-fixed frame angular velocity $\omega$ and the linear acceleration $a$ are continuously measurable. Moreover, we assume that the body frame linear velocity $v_m = R\T v$ is available for measurement. 
Consider a family of $N\geq 1$ constant and known inertial vectors, denoted by $r_i\in \mathbb{R}^3$ for all $i\in  \{1,2,\cdots,N\}$. The measurements of the inertial vectors expressed in the body frame are modeled as
\begin{equation}
	b_i  = R \T r_i,\quad \forall   i\in \{1,2,\cdots,N\}.  \label{eqn:def_b_i}
\end{equation}

The objective of this work is to design a nonlinear continuous attitude estimation scheme on $SO(3)$ for system \eqref{eqn:dynamics_R} with AGAS guarantees in terms of the continuous measurements $y_1=(\omega,a)$ and the intermittent measurements $y_2=(v,r_1,r_2,\dots,r_N)$.

%%%%%%%%%%%%%%%%%%%%%%%%%%%%%%%%%%%%%%%%%%%%%%%%%%%%%%%%%%%%%%%%%%%%%%%%%%%%%%%%
\section{Main Results}

\subsection{Observer Design Using Continuous Measurements}
In this subsection, we consider the case that all the measurements are continuous. Let $\hat{R}\in SO(3), \hat{v}, \hat{g}\in \mathbb{R}^3$ denote the estimates of the attitude, linear velocity and gravity direction, respectively. We propose the following continuous observer on the manifold $SO(3)\times \mathbb{R}^6$:
\begin{subequations} \label{eqn:Cont-Obsv}
	\begin{align}
		\dot{\hat{R}} & = \hat{R}(\omega + k_{o}\hat{R}\T \sigma_R )^\times \label{eqn:Cont-Obsv-R}\\
		\dot{\hat{v}} & = k_o \sigma_R^\times \hat{v} + \hat{g} +  \hat{R} a+ k_v(\hat{R}v_m-\hat{v}) \label{eqn:Cont-Obsv-v} \\%-\omega^\times \hat{v}+ \hat{R}\T\hat{g} +   a + k_v(v-\hat{v}) \\
		\dot{\hat{g}} & = k_o \sigma_R^\times \hat{g} + k_g(\hat{R}v_m-\hat{v})  \label{eqn:Cont-Obsv-g}
	\end{align}
\end{subequations}
where $k_o, k_v, k_g >0 $ and the innovation term $\sigma_R$ is designed as
\begin{align}
	\sigma_R =  -\sum_{i=1}^N \rho_i  r_i^\times \hat{R}  {b}_i - \rho_{N+1} g^\times \hat{g} \label{eqn:def_sigma_R}
\end{align}
with $\rho_i\geq 0$ for all $i\in\{1,2,\dots,N+1\}$.  
Note that the dynamics of $\hat{g}$ designed in \eqref{eqn:Cont-Obsv-g}, together with \eqref{eqn:Cont-Obsv-v}, ensure that $\hat{g}$  converges exponentially to $\hat{R} R\T g$, which allows to consider the gravity direction $g$  (known in the inertial frame) as an additional inertial vector in the design of the innovation term $\sigma_R$.

Note also that the proposed observer   \eqref{eqn:Cont-Obsv} generalizes two existing architectures for the attitude estimation. In particular, if there exists at least two non-collinear inertial vectors, the resulting observer \eqref{eqn:Cont-Obsv-R} with $\rho_{N+1} = 0$ coincides with the nonlinear complementary filter proposed in \cite{mahony2008nonlinear}, \ie, 
\begin{equation} \textstyle 
	\dot{\hat{R}}   = \hat{R}(\omega + k_{o}\hat{R}\T  \sum_{i=1}^N \rho_i   (\hat{R}  {b}_i)^\times r_i )^\times. \label{eqn:CF}
\end{equation}
If the measurements of the linear velocity $v_m$ are available, selecting $\rho_{N+1} >0$ leads to a velocity-aided attitude observer, 
which handles applications with non-negligible linear accelerations, where the accelerometer does not provide the body frame measurements of the gravity vector.

Define the attitude estimation error $\tilde{R} = R\hat{R}\T$, and $\zeta = [\tilde{v}\T,\tilde{g}\T]\T\in \mathbb{R}^6$ with $\tilde{v} =  v- R\hat{R}\T\hat{v} $ and $\tilde{g} = g- R \hat{R}\T\hat{g} $. From \eqref{eqn:psiAR}, \eqref{eqn:def_b_i} and the fact  $  \hat{g} = \tilde{R}\T (g -  \tilde{g})$, the innovation term $\sigma_R$ defined in \eqref{eqn:def_sigma_R}  can be rewritten as
\begin{align}
	\sigma_R &= -\sum_{i=1}^N \rho_i  r_i^\times \tilde{R}\T r_i - \rho_{N+1} g^\times\tilde{R}\T  (g-  \tilde{g})  \nonumber\\ 
	&= \psi(Q\tilde{R}) + \rho_{N+1} g^\times  \tilde{R}\T \tilde{g}\nonumber \\
	&= \psi(Q\tilde{R}) + \Gamma(\tilde{R}) \zeta \label{eqn:sigma_R_error}
\end{align}
where the matrix $Q$ is defined as
\begin{align}
	Q:=\sum_{i=1}^N\rho_i r_i r_i\T + \rho_{N+1} gg\T \in \mathbb{R}^{3\times 3} \label{eqn:def_Q}  
\end{align}
and $\Gamma(\tilde{R}):= [0_{3\times 3},\rho_{N+1} g^\times  \tilde{R}\T]\in \mathbb{R}^{3\times  6}$. It is clear that $\|\Gamma(\tilde{R})\|_F   =  \sqrt{2}\rho_{N+1} \|g\| $ for all $\tilde{R}\in SO(3)$.
 
\begin{lem}\label{lem:Q}
	Consider the matrix $Q$ defined in \eqref{eqn:def_Q} with $\rho_i>0, \forall i\in\{1,2,\dots,N+1\}$. Then, the matrix $\bar{Q} := \tr(Q)I_3-Q$  is positive definite if one of the following statements holds:
	\begin{itemize}
		\item [1)] $N\geq 2$ and there exist at least two non-collinear inertial vectors.
		\item [2)] $N\geq 1$ and there exists at least one inertial vector, which is non-collinear to the gravity vector $g$.
	\end{itemize}
\end{lem} 
	The proof of Lemma \ref{lem:Q} can be easily conducted from \cite[Lemma 2]{tayebi2013inertial} using the fact that the matrix $\bar{Q}$ can be explicitly rewritten as  $\bar{Q}   =   -\sum_{i=1}^N\rho_i (r_i^\times)^2  - \rho_{N+1} (g^\times)^2$ from \eqref{eqn:def_Q}. Moreover, under Lemma \ref{lem:Q} it is always possible to tune the scalar $\rho_i>0$ for all $i\in \{1,2,\dots,N+1\}$ such that the positive definite matrix $\bar{Q}$ has three distinct eigenvalues.

From \eqref{eqn:dynamics_R}, \eqref{eqn:def_a},  \eqref{eqn:Cont-Obsv} and \eqref{eqn:sigma_R_error}, one obtains the following closed-loop system: 
\begin{subequations}\label{eqn:Cont-Closed-Loop}
	\begin{align}
		\dot{\tilde{R}} & = \tilde{R}( -k_o\psi(Q\tilde{R}) - k_o\Gamma(\tilde{R})\zeta )^\times \label{eqn:Cont-Closed-Loop-R}\\
		\dot{\zeta} & = (A-KC)\zeta  \label{eqn:Cont-Closed-Loop-zeta}
	\end{align}
\end{subequations}
with matrices $A,K,C$ given as
\begin{align}
	A = \begin{bmatrix}
		0_3 & I_3\\
		0_3 & 0_3
	\end{bmatrix}, C = \begin{bmatrix}
	I_3 & 0_3
\end{bmatrix}, K=\begin{bmatrix}
	k_vI_3 \\ k_g I_3
\end{bmatrix} \label{eqn:A-C-K}.
\end{align}
Note that the closed-loop system \eqref{eqn:Cont-Closed-Loop} evolves on the manifold $SO(3) \times \mathbb{R}^6$, and one can easily verify that $(I_3,0)$ is one of the equilibrium of system \eqref{eqn:Cont-Closed-Loop}.

\begin{prop}\label{prop:AISS-SO3}
	Let $\mathcal{D}_u$ be a closed and bounded subset of $\mathbb{R}^m$.  Consider the system 	
	\begin{align}
		\dot{\tilde{R}}= \tilde{R} (  -k_o\psi(Q\tilde{R}) + \Gamma(\tilde{R}) u  ) ^\times  \label{eqn:dot_tildeR}
	\end{align}
	with state $\tilde{R}\in SO(3)$, input $u\in \mathcal{D}_u$ and $k_o>0$. Suppose that $\bar{Q} = \tr(Q)I_3 -Q$ is positive definite with three distinct eigenvalues, and there exists a constant $c_\Gamma>0$ such that $\|\Gamma(X_1)-\Gamma(X_2)\|_F \leq c_\Gamma\|X_1-X_2\|_F$ for all $X_1,X_2\in \mathbb{R}^{3\times 3}$.  
	Then, system   \eqref{eqn:dot_tildeR} is almost globally ISS with respect to the equilibrium  $I_3$ and input $u$. 
\end{prop}
\begin{proof} 
	See Appendix \ref{sec:AISS-SO3}
\end{proof}
It is worth to point out that Proposition \ref{prop:AISS-SO3} implies that the nonlinear complementary filter \eqref{eqn:CF} proposed in \cite{mahony2008nonlinear} is almost globally ISS with respect to $I_3$ and some bounded disturbance. The key of the proof of Proposition \ref{prop:AISS-SO3} relies on the fact that system \eqref{eqn:dot_tildeR} (with $u\equiv 0$) is AGAS and has exponentially unstable isolated equilibria \cite{angeli2010stability}. A similar result on almost global ISS of system \eqref{eqn:dot_tildeR}, with $Q = I_3$ and some high gain $k_o$ depending on the bound of the input $u$, can be found in \cite{vasconcelos2011combination} using a combination of Lyapunov and density functions.

\begin{thm}\label{thm:conti}
	Consider the closed-loop system \eqref{eqn:Cont-Closed-Loop} with \eqref{eqn:A-C-K}. Choose the gain parameters as $k_o,k_v,k_g>0$. Then, the equilibrium  $(I_3,0)$ of system \eqref{eqn:Cont-Closed-Loop} is AGAS.
\end{thm}
\begin{proof}
See Appendix \ref{sec:thm1}
\end{proof}

\begin{rem}
	The stability analysis of system \eqref{eqn:Cont-Closed-Loop} relies on the results of Lemma \ref{lem:AGAS_conti} and Proposition \ref{prop:AISS-SO3}.
	Note that the observer \eqref{eqn:Cont-Obsv} can be reduced to 
	\begin{subequations} \label{eqn:Cont-Obsv2}
		\begin{align}
			\dot{\hat{R}} & = \hat{R}(\omega + k_{o}\hat{R}\T \sigma_R )^\times \label{eqn:Cont-Obsv-R2}\\
			\dot{\hat{v}} & = k_o \sigma_R^\times \hat{v}  +  \hat{R} a+ k_v(\hat{R}v_m-\hat{v}) 
		\end{align}
	\end{subequations}
	where the innovation term $\sigma_R$ is given in the same form of \eqref{eqn:def_sigma_R} with $\hat{g} = k_v(\hat{R}v_m-\hat{v})$. 
	 Letting $\tilde{g}   = g-R\hat{R}\T \hat{g}$, one can show that $\tilde{g}$ converges globally exponentially to 0 (\ie, $\dot{\tilde{g}} = -k_v \tilde{g}$) and  $\sigma_R  
	 = \psi(Q\tilde{R}) + \rho_{N+1} g^\times \tilde{R}\T \tilde{g}$. Therefore, AGAS for the reduced observer \eqref{eqn:Cont-Obsv2} is also guaranteed using the similar steps as in the proof of Theorem \ref{thm:conti}. Note also that the reduced observer \eqref{eqn:Cont-Obsv2} has a similar form as in \cite{hua2016stability}. The main drawback of these observers is that the noisy measurements of the linear velocity $v_m$ appear directly in the dynamics of $\hat{R}$ through the innovation term $\sigma_R$.

\end{rem}

\subsection{Observer Design Using Intermittent Measurements}
In practical applications, inertial vector measurements and velocity measurements are often obtained at much lower rates with respect to the IMU measurements. This motivates us to redesign the previous continuous-time observer in terms of  intermittent inertial vectors and linear velocity measurements. In this case, the measurements of the inertial vectors and the linear velocity are available at some time instants $\{t_k\}_{k\in \mathbb{N}_{>0}}$.
\begin{assum}\label{assum:intermittent}
	The time sequence $\{t_k\}_{k\in \mathbb{N}_{>0}}$ is strictly increasing and there exist two constants $0<T_m \leq T_M$ such that $0 \leq t_{1}  \leq T_M$ and $T_m \leq t_{k+1}-t_k \leq T_M, \forall k\in \mathbb{N}_{>0} $. 
\end{assum}
Note that in the particular case where $T_m=T_M=T$, the sampling is periodic with a regular sampling period $T$.

Let $\hat{r}_i, i\in\{1,\dots,N\}$ be the estimate of the inertial vector $r_i$. We propose the following hybrid attitude observer on manifold $SO(3)\times \mathbb{R}^{3N+6}$: 
\begin{align} \arraycolsep=2.9pt\def\arraystretch{1.1}
		\underbrace{\begin{array}{ll}
				\dot{\hat{R}} & = \hat{R}(\omega + k_o\hat{R}\T \sigma_R)^\times \\
				\dot{\hat{v}} & =  k_o \sigma_R^\times \hat{v} + \hat{g} +  \hat{R} a  \\
				\dot{\hat{g}} & = k_o \sigma_R^\times \hat{g} \\				
				\dot{\hat{r}}_i & = k_o \sigma_R^\times \hat{r}_i 
		\end{array}}_{t \in [t_k, t_{k+1}],~    k\in \mathbb{N} } 
		\underbrace{\begin{array}{ll}
				\hat{R}^+ & = \hat{R} \\
				\hat{v}^+ & =   \hat{v}+    k_v(\hat{R}v_m-\hat{v}) \\
				\hat{g}^+ & =  \hat{g} + k_g(\hat{R}v_m-\hat{v}) \\				
				\hat{r}_i^+ & =  \hat{r}_i +    k_{r} (\hat{R}{b}_i-\hat{r}_i)
		\end{array}}_{t \in \{t_k\},~  k\in \mathbb{N} }  
	\label{eqn:hybrid_observer1}
\end{align}
for all $i\in \{1,\dots,N\}$, where $k_o,k_v,k_r,k_g>0$ and the innovation term $\sigma_R$ is designed as
\begin{align}
	\sigma_R =  -\sum_{i=1}^N \rho_i  r_i^\times \hat{r}_i  - \rho_{N+1} g^\times \hat{g} \label{eqn:def_sigma_R2}
\end{align}
with $\rho_i\geq 0$ for all $i\in\{1,2,\dots,N+1\}$. Note that the estimated states $\hat{v}, \hat{g}$ and $\hat{r}_i$ are continuously updated through integration using the continuous angular velocity and linear acceleration measurements and discreetly updated upon the arrival of the intermittent linear velocity and vector measurements. Moreover, it is clear that the estimated attitude $\hat{R}$ from \eqref{eqn:hybrid_observer1} is continuous (not necessary differentiable).  

To capture the behavior of the event-triggered system \eqref{eqn:hybrid_observer1}, a virtual timer $\tau$, motivated from \cite{ferrante2016state,wang2020nonlinear}, is considered with the following hybrid dynamics:
\begin{equation}
	\begin{cases}
		\dot{\tau} ~~= -1, & \tau\in [0,T_M] \\
		\tau^+ \in [T_m, T_M], & \tau\in \{0\}
	\end{cases} \label{eqn:tau}
\end{equation}
with $\tau(0,0)\in [0,T_M]$. Note that the virtual timer $\tau$ decreases to zero continuously, and upon reaching zero it is automatically reset to a value, between $T_m$ and $T_M$, which represents the arrival time of next measurements. %indicating the arrival time next measurements??
With this additional state $\tau$, the time-driven sampling events can be described as state-driven events, which results in an autonomous hybrid closed-loop system.

Let $\bar{\zeta} = [\zeta\T,\tilde{r}_1\T,\dots,\tilde{r}_N\T]\T \in \mathbb{R}^{3N+6}$ with $\tilde{r}_i =  r_i- R\hat{R}\T\hat{r}_i$ for each $i\in \{1,2,\dots,N\}$.  
From \eqref{eqn:psiAR}, \eqref{eqn:def_b_i} and the fact $  \hat{g} = \tilde{R}\T (g -  \tilde{g})$, the innovation term $\sigma_R$ defined in \eqref{eqn:def_sigma_R2} can be rewritten as
\begin{align}
	\sigma_R &= -\sum_{i=1}^N \rho_i  r_i^\times \tilde{R}\T (r_i-\tilde{r}_i) - \rho_{N+1} g^\times\tilde{R}\T  (g-  \tilde{g})  \nonumber\\ 
	&= \psi(Q\tilde{R}) + \sum_{i=1}^N \rho_i r_i^\times \tilde{R}\T \tilde{r}_i + \rho_{N+1} g^\times  \tilde{R}\T \tilde{g} \nonumber \\
	& = \psi(Q\tilde{R}) + \bar{\Gamma}(\tilde{R}) \bar{\zeta}  \label{eqn:sigma_R_error2} 
\end{align}
with
$
	\bar{\Gamma}(\tilde{R}) := [\Gamma(\tilde{R}), \rho_1 r_1^\times \tilde{R}\T, \dots,\rho_N r_N^\times \tilde{R}\T]\in \mathbb{R}^{3\times (3N+6)}
$ and  $Q$ is defined in \eqref{eqn:def_Q}.
It is also clear that $\|\bar{\Gamma}(\tilde{R})\|_F   =  \sqrt{2}\rho_{N+1} \|g\| + \sum_{i=1}^N \sqrt{2} \rho_i \|r_i\|$ for all $\tilde{R}\in SO(3)$.

For the sake of simplicity, let us define the new state $\zeta':= (\bar{\zeta},\tau) \in \mathbb{R}^{3N+6}   \times [0,T_M]$.  From \eqref{eqn:dynamics_R}, \eqref{eqn:def_a},  \eqref{eqn:hybrid_observer1}, \eqref{eqn:tau} and \eqref{eqn:sigma_R_error2}, one obtains the following hybrid closed-loop system:
%\begin{subequations}
\begin{align}\arraycolsep=2.9pt\def\arraystretch{1.1}	\label{eqn:Hybrid-Closed-Loop}
	\begin{cases}
		\left.\begin{array}{ll}
				\dot{\tilde{R}} ~~& = \tilde{R}( -k_o \psi(Q\tilde{R}) -k_o \bar{\Gamma}(\tilde{R}) \bar{\zeta} )\\
				\dot{{\zeta}'} ~~& = \begin{bmatrix}
				\bar{A}\bar{\zeta}  \\
				-1
				\end{bmatrix} 
		\end{array}\right\} & {x\in \mathcal{F}} \\
		\left.\begin{array}{ll}
				\tilde{R}^+ & = \tilde{R} \\
				{\zeta}'^+ & \in \begin{bmatrix}
				 (I_{3N+6} - \bar{K}\bar{C})\bar{\zeta}   \\
				[T_m, T_M]
				\end{bmatrix} 
		\end{array} \qquad \quad ~ \right\} & {x\in \mathcal{J}}  
	\end{cases}
\end{align}
%\end{subequations}
where the flow and jump sets are defined as 
	     $\mathcal{F}  : = SO(3)\times \mathbb{R}^{3N+6} \times [0,T_M], 
		\mathcal{J} :=  SO(3)\times \mathbb{R}^{3N+6} \times \{0\}$,
 
and matrices $ \bar{A}, \bar{C}, \bar{K}$ are given as
\begin{align}
	\bar{A} &= \begin{bmatrix}
			A & 0_{6\times 3N}\\
			0_{3N\times 6} & 0_{3N}
		\end{bmatrix}, 
	\bar{C} = \begin{bmatrix}
		C & 0_{3\times 3N}\\
		0_{3N\times 6} & I_{3N}
	\end{bmatrix}, \nonumber \\
	\bar{K} & =  \begin{bmatrix}
		K & 0_{6\times 3N}\\
		0_{3N\times 3} & k_r I_{3N}
	\end{bmatrix}  \label{eqn:bar-A-C-K}
\end{align} 
with matrices  $ {A}, {C}, {K}$ defined in \eqref{eqn:A-C-K}. Note that $\mathcal{F} \cup \mathcal{J} = SO(3)\times \mathbb{R}^{3N+6} \times [0,T_M]$ and the hybrid closed-loop system {\eqref{eqn:Hybrid-Closed-Loop}} is autonomous and satisfies the hybrid basic conditions of \cite[Assumption 6.5]{goebel2012hybrid}.
 
Now, one can state the following result:
\begin{thm}\label{thm:hybrid_observer}
	Consider the hybrid closed-loop system \eqref{eqn:Hybrid-Closed-Loop}. 	Suppose that Assumption \ref{assum:intermittent} holds. Choose $k_o>0$ and $k_v,k_g,k_r>0$ such that there exists a symmetric positive definite matrix $P$ satisfying
	\begin{equation}
		A_g\T e^{\bar{A}\T \tau} Pe^{\bar{A}\tau}A_g - P<0, \quad \forall \tau\in [T_m, T_M] \label{eqn:P}
	\end{equation}
	with $A_g  := I_{3N+6}-\bar{K}\bar{C}$ and   $\bar{A},\bar{C},\bar{K}$ defined in \eqref{eqn:bar-A-C-K}.
	Then, the set   $\mathcal{A}: = \{I_3\} \times \{0_{(3N+6)\times 1}\} \times [0,T_M]$ is AGAS for the hybrid closed-loop system \eqref{eqn:Hybrid-Closed-Loop}. 
\end{thm}
\begin{proof}
	See Appendix \ref{sec:hybrid_observer}.	
\end{proof}
The optimization problem \eqref{eqn:P} can be solved using
the polytopic embedding technique proposed in \cite{ferrante2016state} and the finite-dimensional LMI approach proposed in \cite{sferlazza2018time}. An explicit procedure motivated from \cite{sferlazza2018time} can be found in \cite{wang2020nonlinear}. Note that this procedure only provides an algorithm to verify the existence of such a symmetric positive definite matrix $P$ satisfying \eqref{eqn:P} when the matrix $\bar{K}$ is properly chosen. However, it is still not clear how to find such gain parameters $k_v,k_g,k_r$, and manual trial‐and‐errors are required in practice. 
The following proposition provides a sufficient condition for the gain parameters $k_v,k_g,k_r$ to guarantee the existence of a solution of \eqref{eqn:P}.
\begin{prop}\label{prop:relaxed-K}
	Let  
	\begin{equation}
		\begin{cases}
			0< k_r <1 \\
			0<k_v <1  \\
			0<k_g<\frac{1-\sqrt{1-k_v}}{T_M} 
		\end{cases} \label{eqn:k_v-k_g-k_r}
	\end{equation}
	Then, there exists a symmetric positive definite matrix $P$ satisfying \eqref{eqn:P}.
\end{prop}
\begin{proof}	
	See Appendix \ref{sec:relaxed-K}
\end{proof}

%%%%%%%%%%%%%%%%%%%%%%%%%%%%%%%%%%%%%%%%%%%%%%%%%%%%%%%%%%%%%%%%%%%%%%%%%%%%%%%%	
\section{Simulation}
In this simulation, we consider an autonomous vehicle equipped with an IMU (including an accelerometer, a gyroscope and a magnetometer) and a Doppler Velocity Log (DVL) sensor providing the linear velocity in the body frame. The accelerometer and gyroscope measurements are sampled at $400(Hz)$, and  the magnetometer and DVL measurements are sampled at about $10(Hz)$ with $T_m=0.09(s)$ and $T_M = 0.11(s)$. An example of the solution of the timer $\tau$ defined in \eqref{eqn:tau} is shown in Fig. \ref{fig:tau}. The vehicle is stabilized along an ``8"-shape trajectory with inertial frame linear velocity  given by $v(t)=[-\sin( t), -4 \sin( t)\cos( t), 0]\T (m/s)$ and angular velocity given by $\omega(t) = [\sin(0.1\pi t),0.1,\cos(0.1\pi t)]\T$. The earth magnetic field and gravity in the inertial frame are given as $r_1 = [0.36,0.64,0]\T$ and $g = [0,0,-9.81]\T$, respectively. For
comparison purposes,  we also consider the continuous observer \eqref{eqn:Cont-Obsv} running at $400(Hz)$ with a zero-order-hold (ZOH) method when the measurements of the linear velocity and the inertial vectors are not available. 

\begin{figure}[!ht]
	\centering
	\includegraphics[width=0.95\linewidth]{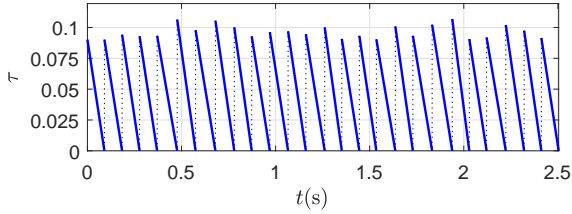} \vspace*{-0.26cm}
	\caption{The evolution of $\tau$ in \eqref{eqn:tau} with $T_m=0.09$ and $T_M = 0.11$.}
	\label{fig:tau}
\end{figure}

The initial conditions are chosen as $R(0) = I_3$, $\hat{R} = \mathcal{R}_a(0.99\pi, u), u\in\mathbb{S}^2$, and $\hat{v}(0)=\hat{g}(0)=\hat{r}_1(0)= 0$. The gain parameters are tuned such that both observers have similar convergence rate with $k_o = 15, k_v=2.5, k_g=8$ for observer \eqref{eqn:Cont-Obsv} and $k_o = 15, k_v=0.7,k_g=4,k_r=0.1$ for observer \eqref{eqn:hybrid_observer1}.  Two sets of  simulation results are shown in Fig. \ref{fig:Simulation}. The first case considers noise-free measurements, while the second case considers the measurements corrupted with zero mean Gaussian noise of $0.01$ variance in the gyro and magnetometer measurements and $0.1$ variance in the accelerometer and DVL measurements.  As one can see, the steady state attitude estimation error of our hybrid observer \eqref{eqn:hybrid_observer1} is significantly less than that of the continuous observer \eqref{eqn:Cont-Obsv} with a ZOH method. It is worth pointing out that, even in the noise-free case,  the steady state estimation errors of the continuous observer \eqref{eqn:Cont-Obsv} with a practical ZOH method do not converge to zero in the presence of intermittent measurements.

	\begin{figure}[!ht]
		\centering 
			\includegraphics[width=0.9\linewidth]{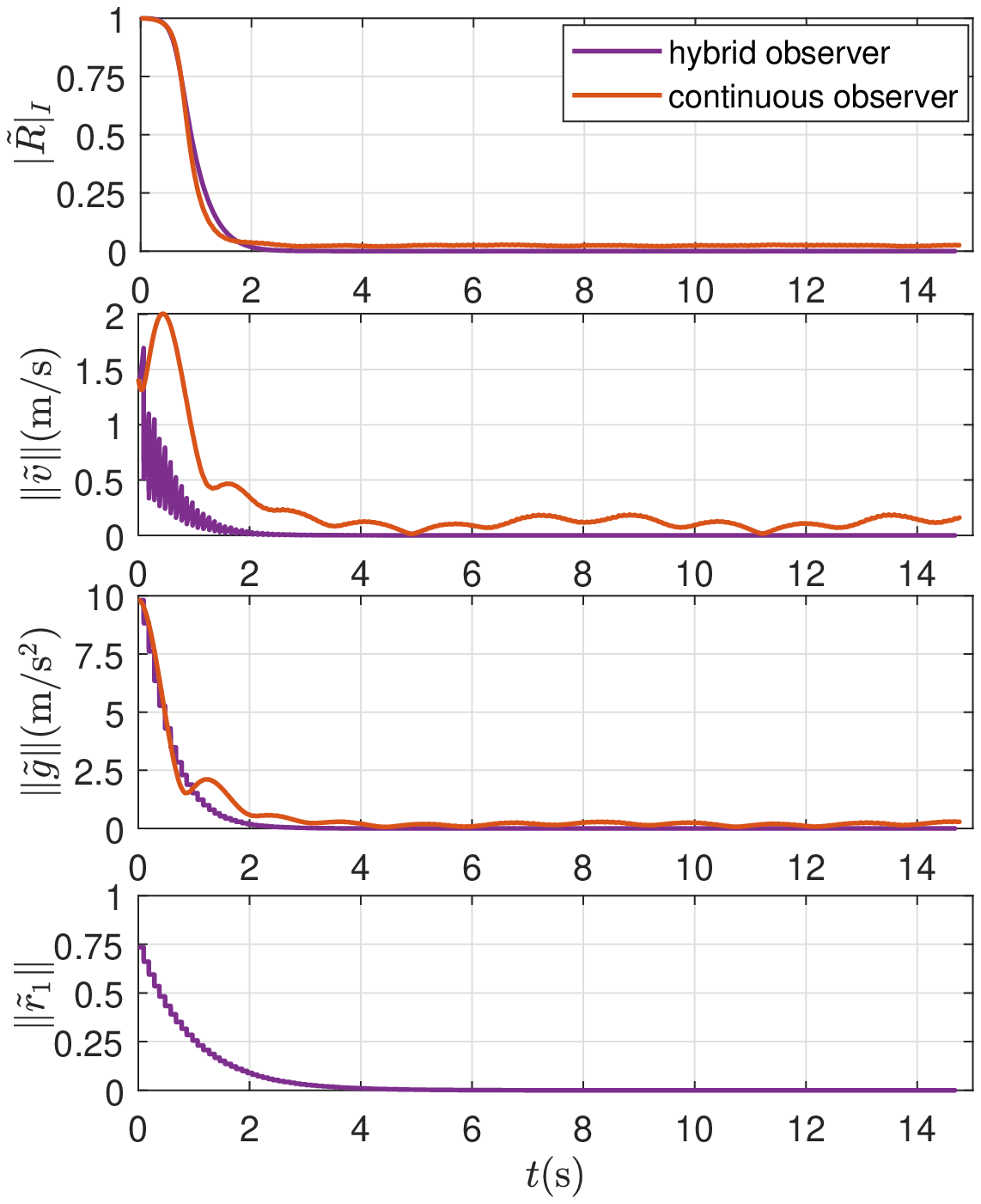}  \\ 
			\includegraphics[width=0.9\linewidth]{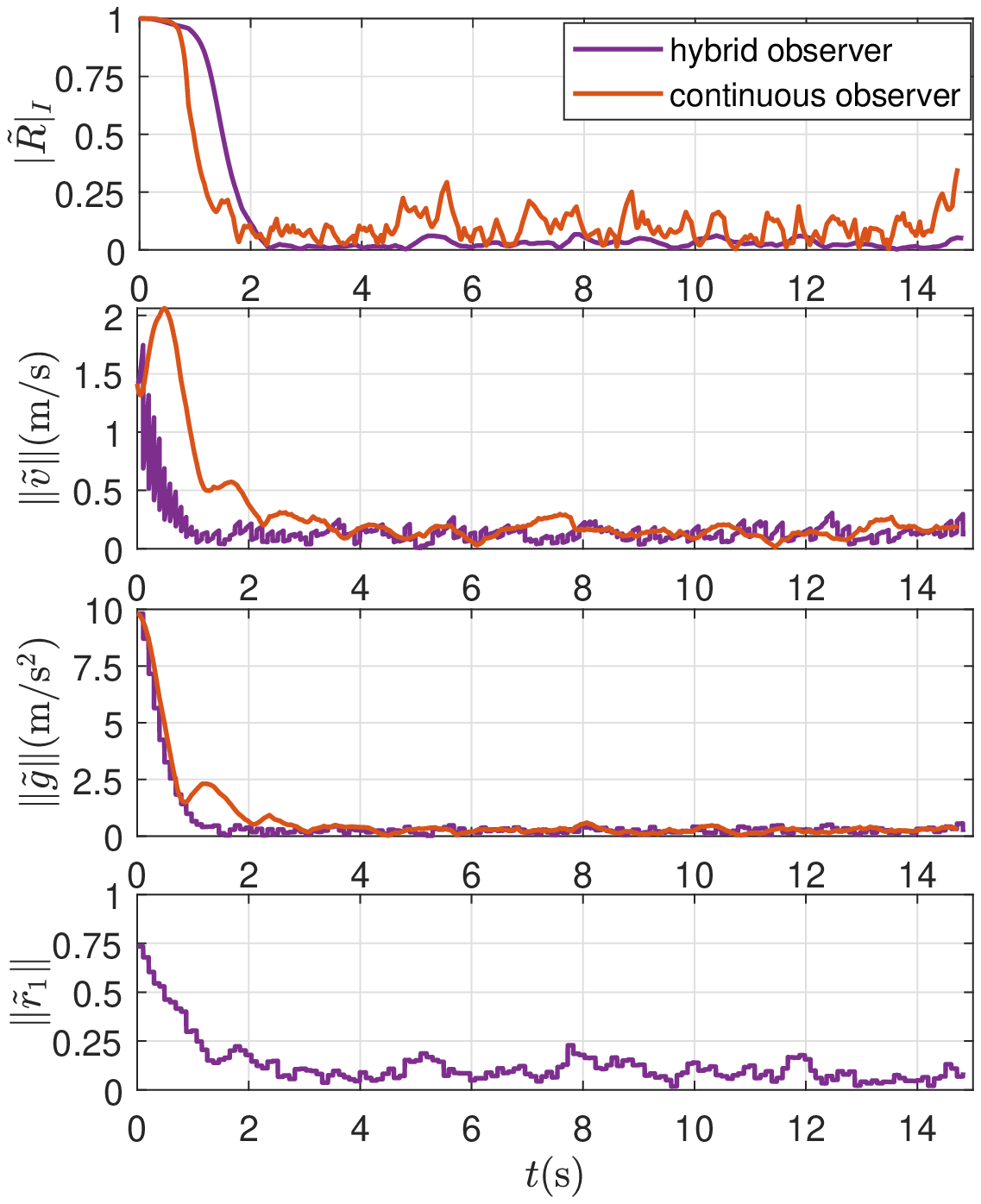} 
        \vspace*{-0.28cm}
		\caption{Simulation results of continuous observer \eqref{eqn:Cont-Obsv} with ZOH   and hybrid observer \eqref{eqn:hybrid_observer1}. The results with noise-free measurements and noisy measurements are shown in the first and second plot, respectively.}
		\label{fig:Simulation}
	\end{figure}  

%%%%%%%%%%%%%%%%%%%%%%%%%%%%%%%%%%%%%%%%%%%%%%%%%%%%%%%%%%%%%%%%%%%%%%%%%%%%%%%%	
\section{Conclusion}
In this work, we formulated the problem of velocity aided attitude estimation problem with intermittent measurements as an interconnection of an almost global ISS system and a GAS system, and we proved that the interconnected system is almost globally asymptotically stable. It is shown that the first version of our observer, relying on continuous measurements, does not preserve its theoretically guaranteed convergence and performance when the measurements are intermittent in nature. This remark is true for the available velocity-aided observers existing in the literature. To overcome this problem, we redesigned our attitude observer using hybrid systems tools to efficiently handle the case where the measurements of the linear velocity and inertial vectors are intermittent. We introduced a virtual hybrid counter to capture the intermittent nature of the measurements and proposed a hybrid velocity-aided attitude observer with AGAS guarantees. The simulation results show that this hybrid observer exhibits better performance than the observer designed with continuous measurements, when implemented in applications involving sensors with different bandwidth properties.

%%%%%%%%%%%%%%%%%%%%%%%%%%%%%%%%%%%%%%%%%%%%%%%%%%%%%%%%%%%%%%%%%%%%%%%%%%%%%%%%
%\section*{Appendix}

\appendix

\subsection{Proof of Proposition \ref{prop:AISS-SO3}} \label{sec:AISS-SO3}  
The proof of   Proposition \ref{prop:AISS-SO3} relies on  the results in \cite[Proposition 2]{angeli2010stability}. We first show that system \eqref{eqn:dot_tildeR} satisfies the three Assumptions A0–A2 in \cite{angeli2010stability}.  
One can easily verify that A0 is fulfilled, since system \eqref{eqn:dot_tildeR} evolves on the compact manifold $SO(3)\times \mathcal{D}_u$ and $\|\Gamma(X_1)-\Gamma(X_2)\|_F \leq c_\Gamma\|X_1-X_2\|_F$. 
Moreover, consider the smooth function on $SO(3)$
\begin{equation}
	V(\tilde{R}) = \tr(Q(I_3-\tilde{R})) \label{eqn:V_R}
\end{equation}
whose time derivative along the trajectory \eqref{eqn:dot_tildeR} with $u\equiv 0$ is given as 
\begin{align*}
	\dot{V}(\tilde{R}) &= \tr(-Q\tilde{R}(-k_o \psi(Q\tilde{R})^\times))   = -2 k_o \|\psi(Q\tilde{R})\|^2.
\end{align*}
This implies that  $\dot{V}<0$ for all $\tilde{R}\in SO(3)$ and $\psi(Q\tilde{R})\neq 0$, and then A1 is fulfilled. Applying LaSalle’s principle, it follows that the solution $\tilde{R}$ of system \eqref{eqn:dot_tildeR} with $u\equiv 0$  converges asymptotically to the set of equilibria $\mathcal{W} := \{\tilde{R}\in SO(3): \psi(\bar{Q}\tilde{R}) = 0 \}$. Since  $\psi(\bar{Q}\tilde{R}) = 0$ implies $\bar{Q}\tilde{R} = \tilde{R}\T \bar{Q}$, one can explicitly rewrite $\mathcal{W}$ as $\mathcal{W} = \{I_3\}\cup \{\tilde{R}\in SO(3): \tilde{R} = \mathcal{R}_a(\pi,v), v\in \mathcal{E}(\bar{Q})\}$. Note that the undesired equilibria in the set $\{\tilde{R}\in SO(3): \tilde{R} = \mathcal{R}_a(\pi,v), v\in \mathcal{E}(\bar{Q})\}$ are isolated since $\bar{Q}$ has three distinct eigenvalues. Moreover, one can show that the linearized system (with $u\equiv 0$) at each undesired equilibrium has at least one positive eigenvalue (for instance, see the proof of \cite[Theorem 1]{wang2021nonlinear}).  Hence, the equilibrium point $I_3$ of \eqref{eqn:dot_tildeR} with $u\equiv 0$ is almost globally asymptotically stable and system \eqref{eqn:dot_tildeR} satisfies Assumptions A0–A2 in \cite{angeli2010stability}.

On the other hand, from \eqref{eqn:dot_tildeR} one can show that
\begin{align}
	\frac{d}{dt}|\tilde{R}|_I^2 &= \frac{1}{2}\psi\T(\tilde{R}) ( -  k_o  \psi(Q\tilde{R}) +  \Gamma(\tilde{R})u ) \nonumber \\
	&\leq  -\frac{1}{2}k_o\lambda_m^{\bar{Q}} \|\psi(\tilde{R})\|^2  +   \frac{1}{2}  \|\psi(\tilde{R})\| \|\Gamma(\tilde{R}) u\| \nonumber \\
	&\leq  -2k_o\lambda_m^{\bar{Q}} (1-|\tilde{R}|_I^2)|\tilde{R}|_I^2  +   \frac{\sqrt{3}c_\Gamma}{4}    \|u\|  \nonumber \\
	& \leq  -2k_o\lambda_m^{\bar{Q}} |\tilde{R}|_I^2 + 2k_o\lambda_m^{\bar{Q}}   +     \frac{\sqrt{3} c_\Gamma  c_u }{4}   \label{eqn:dot_tilde_R_I1}
\end{align}
where $c_u: = \max_{u\in \mathcal{D}_u}  \|u\|$, $|\tilde{R}|_I^2 = \frac{1}{4}\tr(I_3-\tilde{R}) \in [0,1]$, and we made use of the facts  $\psi\T(\tilde{R})\psi(Q\tilde{R}) = \psi\T(\tilde{R})\bar{Q}\psi(\tilde{R}) \leq  \lambda_m^{\bar{Q}}\|\psi(\tilde{R})\|^2$, $\|\Gamma(\tilde{R})\|_F \leq c_\Gamma \|\tilde{R}\|_F = \sqrt{3}c_\Gamma $ and $\|\psi(\tilde{R})\|^2 = 4(1-|\tilde{R}|_I^2)|\tilde{R}|_I^2\leq 1$ for all $\tilde{R}\in SO(3)$.  Hence, by virtue of \cite[Proposition 3]{angeli2010stability},   system \eqref{eqn:dot_tildeR} fulfills the ultimate boundedness property. Therefore, one can conclude from \cite[Proposition 2]{angeli2010stability} that system \eqref{eqn:dot_tildeR} is almost globally ISS with respect to the equilibrium $I_3$ and input $u$.

\subsection{Proof of Theorem \ref{thm:conti}}\label{sec:thm1}
From \eqref{eqn:Cont-Closed-Loop}, the overall closed-loop system has the same structure as the one described in Lemma \ref{lem:AGAS_conti}.
Hence, we first show that the equilibrium ($\zeta = 0$) of $\zeta$-subsystem \eqref{eqn:Cont-Closed-Loop-zeta} with matrices $A,C,K$ defined in \eqref{eqn:A-C-K} is globally exponentially stable. From \eqref{eqn:A-C-K}, one can easily show that
\begin{align*}
	A-KC = \begin{bmatrix}
		-k_v I_3 & I_3\\
		-k_g I_3 & 0_3 
	\end{bmatrix}
\end{align*} 
which implies that matrix $A-KC$ is Hurwitz for all $k_v,k_g>0$. Hence, for each $\mu>0$ there exists a symmetric positive definite matrix $P$ satisfying the following Lyapunov equation 
\begin{equation}
	(A-KC)\T P + P(A-KC)  = -\mu I_6. \label{eqn:Lyapunov_Eq}
\end{equation}
Consider the Lyapunov function candidate $V(\zeta) = \zeta\T P\zeta$. From \eqref{eqn:Cont-Closed-Loop-zeta}, \eqref{eqn:A-C-K} and \eqref{eqn:Lyapunov_Eq}, one can easily show that 
\begin{align}
	\dot{V}(\zeta) &= \zeta\T\left((A-KC)\T P + P(A-KC) \right)\zeta \nonumber \\
	& = - \mu \|\zeta\|^2 \leq -\frac{\mu}{\lambda_M^P} V(\zeta).
\end{align}
It follows that $\zeta$ converges globally exponentially to zero. Moreover, from Proposition \ref{prop:AISS-SO3}, one obtains that the $\tilde{R}$-subsystem \eqref{eqn:Cont-Closed-Loop-R} is almost globally ISS with respect to the equilibrium $I_3$ and input $\zeta$. Therefore,  by virtual of Lemma \ref{lem:AGAS_conti},  the equilibrium  $(I_3,0)$ of   system \eqref{eqn:Cont-Closed-Loop} is AGAS.

\subsection{Proof of Theorem \ref{thm:hybrid_observer}}\label{sec:hybrid_observer} 
From \eqref{eqn:Hybrid-Closed-Loop},   the overall closed-loop system has the same structure as the one described in Lemma \ref{lem:AGAS_hybrid}. 
Hence, similar to the proof of Theorem \ref{thm:conti}, we first show that the set $\mathcal{A}':=\{0_{(3N+6)\times 1}\} \times [0,T_M]$ is globally exponentially stable for the $\zeta'$-subsystem  with matrices $\bar{A},\bar{C},\bar{K}$ defined in \eqref{eqn:bar-A-C-K}. Consider the following Lyapunov function candidate:
\begin{align}
	V(\zeta') = \bar{\zeta}\T e^{\bar{A}\T \tau} Pe^{\bar{A}\tau}\bar{\zeta}  \label{eqn:Lyapunov_Vzeta'}
\end{align}
where $P=P\T >0$ is the solution to \eqref{eqn:P}. Let $|\zeta'|_{\mathcal{A}'}: = \inf_{y\in \mathcal{A}'} \|\zeta'-y\| = \|\bar{\zeta}\| $. One can easily verify that
\begin{equation}
	\underline{\alpha} |\zeta'|_{\mathcal{A}'}^2 \leq 	V(\zeta') \leq  \bar{\alpha} |\zeta'|_{\mathcal{A}'}^2  \label{eqn:Lyapunov_Vzeta'-bound}
\end{equation}
where 
$\underline{\alpha} := \min_{\tau \in [0,T_M]} \lambda_m{(e^{\bar{A}\T \tau} e^{\bar{A}\tau})}  \lambda_m(P)
$ and $\bar{\alpha} := \max_{\tau \in [0,T_M]} \lambda_M{(e^{\bar{A}\T \tau} e^{\bar{A}\tau})}   \lambda_M(P)
$. 
Since the matrix $\bar{A}$, defined in \eqref{eqn:bar-A-C-K},  is nilpotent with $\bar{A}^2=0$, one can verify that $e^{\bar{A}\tau} = \sum_{k=0}^{\infty} \frac{1}{k!}\bar{A}^k =I_{3N+6} + \bar{A}\tau$ %$\|\exp(\bar{A}\tau)\bar{\zeta}\|^2 = \bar{\zeta}\T \exp(\bar{A}\T\tau)\exp(\bar{A}\tau)\bar{\zeta} $ 
and $0< \lambda_m{(e^{\bar{A}\T \tau} e^{\bar{A}\tau})} \leq 1\leq \lambda_M{(e^{\bar{A}\T \tau} e^{\bar{A}\tau})}$ for all $\tau\in [0,T_M]$.  
Using the facts that $\frac{d}{dt}e^{\bar{A}\tau} = \dot{\tau}\bar{A} e^{\bar{A}\tau} = - \bar{A} e^{\bar{A}\tau}$ and $\bar{A}e^{\bar{A}\tau}=e^{\bar{A}\tau} \bar
{A}$, one obtains
\begin{align*}
	\frac{d}{dt}e^{\bar{A}\T\tau} Pe^{\bar{A}\tau}    = e^{\bar{A}\T\tau} (-\bar{A}\T P - P\bar{A})e^{\bar{A}\tau}.
\end{align*}
Thus, the time-derivative of ${V}(\zeta')$ along the flows of \eqref{eqn:Hybrid-Closed-Loop} is given by 
\begin{align}
	\dot{V}(\zeta') &= \dot{\bar{\zeta}}\T e^{\bar{A}\T\tau} Pe^{\bar{A}\tau} \bar{\zeta}   +  \bar{\zeta}\T e^{\bar{A}\T\tau} Pe^{\bar{A}\tau} \dot{\bar{\zeta}} \nonumber \\
	&\quad  +  \bar{\zeta}\T e^{\bar{A}\T\tau}   (-\bar{A}\T P - P\bar{A}) e^{\bar{A}\tau} \bar{\zeta} \nonumber \\
	& = 0, \quad \forall (\tilde{R},\zeta')\in \mathcal{F}.   \label{eqn:Lyapunov_Vzeta'dot}
\end{align}
This implies that $V(\zeta')$ is non-increasing  in the flows.  
Since inequality \eqref{eqn:P} holds for all $\tau\in [T_m,T_M]$,  there exists a (small enough) positive scalar $\beta <  \bar{\alpha}$ such that  
\[
A_g\T e^{\bar{A}\T\tau} Pe^{\bar{A}\tau}A_g - P \leq  - \beta I_{3N+6} <0, \forall \tau\in [T_m,T_M]. 
\]
Hence, for each jump it follows  from \eqref{eqn:Hybrid-Closed-Loop} and \eqref{eqn:P}-\eqref{eqn:Lyapunov_Vzeta'-bound} that
\begin{align}
	V(\zeta'^+) & =  \bar{\zeta}\T A_ge^{\bar{A}\T\tau} Pe^{\bar{A}\tau} A_g\bar{\zeta} \nonumber \\
	&= V(\zeta')  + \bar{\zeta}\T(A_g\T e^{\bar{A}\T\tau} Pe^{\bar{A}\tau}A_g - P) \bar{\zeta} \nonumber \\
	&\leq  V(\zeta')-\beta \|\bar{\zeta}\|^2 \nonumber\\
	& \leq \left( 1-\frac{\beta}{\bar{\alpha}}\right)  V(\zeta') \nonumber \\
	&\leq  e^{-\lambda_J } V(\zeta') , \quad  \forall (\tilde{R},\zeta')\in \mathcal{J}   \label{eqn:Lyapunov_Vzeta'+} 
\end{align}
where $\tau\in [T_m,T_M], \lambda_J : = -\ln(1-\frac{\beta}{\bar{\alpha}})$. 
Using the fact $\beta <  \bar{\alpha}$, it is clear that $0<1-\frac{\beta}{\bar{\alpha}} <1$ and $\lambda_J>0$. Hence, $V(\zeta')$ is also non-increasing during the jumps. Note that the hybrid closed-loop system \eqref{eqn:Hybrid-Closed-Loop} satisfies the hybrid basic conditions \cite[Assumption 6.5]{goebel2012hybrid}.  By virtue of  \cite[Definition 2.6]{goebel2012hybrid}, it is straightforward to check that for every initial condition $\zeta'(0,0)\in \mathbb{R}^{3N+6} \times [0,T_M]$ there exists at least a nontrivial solution to \eqref{eqn:Hybrid-Closed-Loop} and that every maximal solution to \eqref{eqn:Hybrid-Closed-Loop} is complete, \ie, $t+j \to +\infty$. Since $T_m$ is strictly positive by Assumption \ref{assum:intermittent}, there is no Zeno behavior and $t\to + \infty$ as  $t+j \to +\infty$. Moreover, using the fact $j T_m \leq t\leq j T_M + T_M$ for each $(t,j)\in \dom \zeta'$, one has   $j\geq \frac{1}{1+T_M} (t+j) - \frac{T_M}{1+T_M}$. Hence, from \eqref{eqn:Lyapunov_Vzeta'}, \eqref{eqn:Lyapunov_Vzeta'dot} and \eqref{eqn:Lyapunov_Vzeta'+}, one can show that
\begin{align*}
	V(\zeta'(t,j)) &\leq  e^{-\lambda_J j} V(\zeta'(0,0)) \nonumber \\
	& \leq    e^{\left( -\frac{\lambda_J }{1+T_M} (t+j) + \frac{\lambda_J T_M}{1+T_M}\right) } V(\zeta'(0,0)) \nonumber \\
	& \leq   \kappa e^{-\lambda (t+j)} V(\zeta'(0,0)) ,\quad  \forall (t,j)\in \dom \zeta'  
\end{align*}
where $\lambda:= \frac{\lambda_J }{1+T_M}$ and $\kappa: = e^{\frac{\lambda_J T_M}{1+T_M}}$. From \eqref{eqn:Lyapunov_Vzeta'-bound}, one can further show that 	$|\zeta'(t,j)|_{\mathcal{A}'}   
	 \leq   \sqrt{ \frac{\bar{\alpha}\kappa}{\underline{\alpha}}}e^{-\frac{\lambda}{2} (t+j)}  |\zeta'(0,0)|_{\mathcal{A}'} $ for all $ (t,j)\in \dom \zeta' $, 
which implies that the sub-state $\zeta'$ of the overall system \eqref{eqn:Hybrid-Closed-Loop} converges globally exponentially to $\mathcal{A}'$. Similar to the proof of Theorem \ref{thm:conti},  by virtual of Proposition \ref{prop:AISS-SO3} and Lemma \ref{lem:AGAS_hybrid}, one concludes that equilibrium set   $\mathcal{A}$ is AGAS for the hybrid closed-loop system \eqref{eqn:Hybrid-Closed-Loop}.

\subsection{Proof of Proposition \ref{prop:relaxed-K}} \label{sec:relaxed-K}
Recall the definitions of $\bar{A},\bar{C},\bar{K}$ defined in \eqref{eqn:bar-A-C-K} and $A_g = I_{3N+6}-\bar{K}\bar{C} $, one obtains 
\begin{align*}
	e^{\bar{A}\tau}A_g =  \begin{bmatrix}
		e^{{A}\tau}(I_6-KC) & 0_{6\times 3N}\\
		0_{3N\times 3} & (1-k_r)I_{3N}
	\end{bmatrix}  
\end{align*} 
with $ {A}, {C}, {K}$ defined in \eqref{eqn:A-C-K}. 
Choosing $P = \diag(\bar{P},I_{3N})$ with some $\bar{P}\in \mathbb{R}^{6\times 6}$, inequality \eqref{eqn:P} holds if 
\begin{equation}
	(1-k_r)^2 < 1 \label{eqn:k_r}
\end{equation} 
and there exists a symmetric matrix $\bar{P}>0$ satisfying 
\begin{multline}
	%	\begin{align} 
	(I_6- {K} {C})\T e^{ {A}\T\tau} \bar{P}e^{{A}\tau}(I_6- {K} {C}) - \bar{P}<0  \label{eqn:barP}
	%	\end{align}
\end{multline}
for all  $\tau\in [T_m, T_M] $.
Applying the discrete-time Lyapunov equation, the existence of $\bar{P}$ satisfying \eqref{eqn:barP} for all  $\tau\in [T_m, T_M] $ is guaranteed if all the eigenvalues of $e^{ {A}\tau}(I_6- {K} {C})$ are located in the unit circle for all $\tau\in [T_m, T_M]$. Using the fact
\begin{align*}
	e^{ {A}\tau}(I- {K} {C}) = \begin{bmatrix}
		(1- k_g \tau  -k_v)I_3 & \tau I_3 \\
		-k_g I_3 & I_3
	\end{bmatrix}
\end{align*}
one can verify that the eigenvalues of the matrix $e^{ {A}\tau}(I_6- {K} {C})$ are in the form of
$
	%	1+ \frac{1}{2}\left( -\sqrt{(k_g \tau + k_v)^2 - 4 k_g \tau} - k_g \tau - k_v\right) \\
	\lambda(\tau) = 1- \frac{ 1}{2}( (k_g \tau + k_v) \pm \sqrt{(k_g \tau + k_v)^2 - 4 k_g \tau  }   ).   %\label{eqn:lambda}
$
To guarantee that all the eigenvalues $\lambda(\tau)$ are in the unit circle for all $\tau\in [T_m, T_M]$, it is sufficient to choose $k_v>0$ and $k_g>0$ satisfying    
\begin{align}
	\begin{cases}
		(k_g \tau + k_v)^2 - 4 k_g \tau\geq 0 \\
		 (k_g \tau + k_v) + \sqrt{(k_g \tau + k_v)^2 - 4 k_g \tau  }   < 2
	\end{cases}  \label{eqn:k_v-k_g}
\end{align} 
for all $ \tau\in [T_m, T_M]$.
One can further show that inequalities \eqref{eqn:k_v-k_g} hold for all $ \tau\in [T_m, T_M]$ if  
\begin{equation}
%	\begin{cases}
		0<k_v <1,\quad 
		0<k_gT_M < 1-\sqrt{1-k_v}  %\text{ or } 1-\sqrt{1-k_v} > k_g T_M
%	\end{cases} 
\label{eqn:k_v-k_g2}.
\end{equation} 
Therefore, one concludes  \eqref{eqn:k_v-k_g-k_r} from     \eqref{eqn:k_r} and \eqref{eqn:k_v-k_g2}.  
	
\bibliographystyle{IEEEtran}
\bibliography{mybib}
	
%%%%%%%%%%%%%%%%%%%%%%%%%%%%%%%%%%%%%%%%%%%%%%%%%%%%%%%%%%%%%%%%%%%%%%%%%%%%%%%%

\end{document}